\documentclass[a4paper,twocolumn,10pt,unpublished]{quantumarticle}
\pdfoutput=1
\usepackage[utf8]{inputenc}
\usepackage[T1]{fontenc}
\usepackage[english]{babel}
\usepackage{amsmath,amssymb,amsthm,mathtools}
\usepackage{graphicx}
\usepackage{booktabs}
\usepackage{tabularx}
\usepackage[numbers,sort&compress]{natbib}
\usepackage{microtype}
\usepackage{placeins}
\usepackage{hyperref}

\providecommand{\tightlist}{%
  \setlength{\itemsep}{0pt}\setlength{\parskip}{0pt}}

\numberwithin{equation}{section}
\newtheorem{theorem}{Theorem}[section]
\newtheorem{lemma}[theorem]{Lemma}
\newtheorem{proposition}[theorem]{Proposition}
\newtheorem{corollary}[theorem]{Corollary}

\theoremstyle{remark}

\begin{document}

\title{Quantum Change Interval: Exact Asymptotics for Minimum Error Localization}

\author{Xu~Chen}
\orcid{0009-0003-8499-4080}
\email{cnc@hebust.edu.cn}
\affiliation{School of Sciences, Hebei University of Science and Technology,
Shijiazhuang, Hebei 050018, People's Republic of China}
\author{Xue~Ma}
\email{maxue@he.chinamobile.com}
\affiliation{Network Management Center, China Mobile Communications Group
Hebei Co., Ltd., Shijiazhuang, Hebei 050000, People's Republic of China}

\maketitle
\hypersetup{
  pdftitle={Quantum Change Interval: Exact Asymptotics for Minimum Error Localization},
  pdfauthor={Xu Chen and Xue Ma}
}

\begin{abstract}
We study a returning quantum change interval, in which a source emits
\(\lvert\psi\rangle\) over one interval and
\(\lvert0\rangle\) elsewhere. A collective measurement on the full sequence identifies both endpoints under the minimum error criterion. We analyze the Gram matrix through Toeplitz comparison and Følner
transfer, together with an exact decomposition by excitation number and
interval hull. The resulting bounds prove that the square root measurement(SRM) attains the Bayes optimum asymptotically. Let \(c=\lvert\langle0\vert\psi\rangle\rvert\) be the local state overlap, and define \(p_1(x)=4(1-x^2)K^2(x^2)/\pi^2\) with \(K\) the complete elliptic integral of the first kind. For a known interval length \(i\), the SRM and Bayes optimal success probabilities converge to the same Toeplitz symbol integral as the number \(N\) of translations grows. Their difference satisfies
\(
P_{\rm opt}(G_{N,i})-P_{\rm SRM}(G_{N,i})
=O_{i,c}(N^{-1/2}).
\)
If both \(i\) and \(N\) diverge, their common limit is \(p_1(c^2)\), with no
constraint on their relative growth. If the interval length is unknown, the uniform prior over all \(M_n=n(n+1)/2\) intervals gives the common fixed overlap limit \(p_1(c)^2\). We derive the uniform scaling law
\[
M_nP_X
=
\left(
1+\frac{2\sqrt{\tau_n}}{\pi}
+\frac{\sqrt{2}\,\tau_n}{\pi^2}
\right)^2
+
O_T\!\left(\frac{\log\log n}{\log n}\right),
\]
\(X\in\{\mathrm{tr},\mathrm{SRM},\mathrm{opt}\}\). More generally, if \(c_n\) approaches one from below and \(n p_1(c_n)\to\infty\), the same
three quantities satisfy \(P_X\sim p_1(c_n)^2\). The SRM is therefore
asymptotically optimal. The asymptotic laws also extend to joint detection and exact localization in the presence of no change prior.

\end{abstract}

\section{Introduction}\label{introduction}

Quantum change point inference asks where an ordered source switches from one
state preparation process to another. A quantum receiver may store the emitted
systems and measure the entire sequence collectively rather than deciding from
each output in turn. Early Bayesian work considered unknown pure states and
two candidate change locations \cite{akimoto2011}. For a permanent change
between known pure states, collective minimum error discrimination has a
nonzero asymptotic success probability and can outperform online strategies
based on local measurements \cite{sentis2016}. Protocols for exact and online
unambiguous identification were developed later
\cite{sentis2017,sentis2018}, and an adaptive
Bayesian protocol was demonstrated with photonic state sequences
\cite{yu2018}. Quickest and sequential quantum change detection instead optimize a false alarm--delay tradeoff
\cite{fanizza2023,guha2025,zecchin2026}. Hamiltonian change points have been
studied as finite quantum process discrimination problems under a
minimum error criterion \cite{nakahira2023}, whereas unambiguous channel
discrimination permits an inconclusive outcome \cite{nakahira2026}. Quantum change point identification has also been studied under local
operations and classical communication, where it is connected to
entanglement distillation \cite{banerjee2024}. General Gram matrix
semidefinite programs cover multiple changes \cite{mohan2023}, while related
localization problems include malfunctioning devices, multiple anomalies, and
boundaries between quantum domains
\cite{skotiniotis2024,llorens2024,llorens2025}.

\begin{table*}[t]
\centering
\caption{Comparison with related quantum and classical models.}
\label{tab:literature}
\footnotesize
\begin{tabularx}{\textwidth}{@{}
>{\raggedright\arraybackslash}p{0.17\textwidth}
>{\raggedright\arraybackslash}p{0.24\textwidth}
>{\raggedright\arraybackslash}X@{}}
\toprule
Work & Relation & Established result and distinction from the present work \\
\midrule
Sentís et al.~\cite{sentis2016,sentis2017,sentis2018}
& Boundary slice \(b=n\)
& minimum error law for a permanent change, with separate unambiguous and
online protocols. A returning interval couples the entry and return
boundaries. \\
Mohan--Sikora--Upadhyay \cite{mohan2023}
& Specialization with two changes after a common leading vacuum factor:
\(N'=n+1,c_1=a,c_2=b+1\), and states
\(\lvert0\rangle,\lvert\psi\rangle,\lvert0\rangle\). Removing the common
factor preserves the Gram matrix; after phase fixing, the pairwise overlaps
are \((\gamma_1,\gamma_2,\gamma_{12})=(c,1,c)\)
& A finite size Bayes SDP based on the Gram matrix and a piecewise kernel. The
present analysis derives fixed overlap and moving overlap asymptotic laws,
including a sufficient hull dominance condition, and proves SRM asymptotic
optimality for the full physical ensemble. \\
Quantum multi-anomaly detection \cite{llorens2024}
& Consecutive \(k=i\) subensemble; strict for \(1<i<n\)
& Johnson symmetry gives finite \(n\) SRM optimality for the full
multi-anomaly ensemble. Its unambiguous law \((1-c^2)^k\) has a local protocol
and is also the asymptotic minimum error limit for fixed \(k\). Contiguity
removes Johnson transitivity. \\
Quantum edge detection \cite{llorens2025}
& Analogue with one boundary and unknown domain states
& An SRM protocol is asymptotically optimal. The returning model has two
jointly inferred endpoints, calibrated states, and equal exterior regions. \\
Classical structured change statistics
\cite{levin1985,yao1993,walther2010,chen2015}
& Epidemic and changed interval models are geometric analogues; spatial scans
and graph based statistics are related structured detection methods
& Scan, CUSUM, and graph based methods use classical significance or power
criteria rather than exact Bayes localization by collective POVMs. \\
\bottomrule
\end{tabularx}
\end{table*}

In this work, the change is transient. The source emits a known anomalous pure
state \(\lvert\psi\rangle\) on one nonempty interval and then
returns to the reference state \(\lvert0\rangle\). We call this returning
structure a \emph{quantum change interval}. For a sequence of length
\(n\) and an interval \(I=[a,b]\), the corresponding hypothesis is
\[
\lvert\Phi_{a,b}\rangle
=\lvert0\rangle^{\otimes(a-1)}
 \lvert\psi\rangle^{\otimes(b-a+1)}
 \lvert0\rangle^{\otimes(n-b)}.
\]
We write \(i=b-a+1\) for the interval length,
\(c=\lvert\langle0\vert\psi\rangle\rvert\) for the local state overlap, and
\(r=c^2\). When the interval length \(i\) is known, there are \(N=n-i+1\) admissible translations
\(a=1,\dots,N\), so only the starting point \(a\) must be inferred.
When the interval length is unknown, both endpoints must be inferred. Figure~\ref{fig:model-geometry}
shows the returning geometry and the collective localization task.
The permanent change model is the boundary slice \(b=n\), while fixed length
intervals form a consecutive anomaly subensemble with tied endpoints. Quantum
edge detection has related boundary geometry, but here both states are
calibrated and the return boundary is explicit.

\begin{figure*}[t]
  \centering
  \includegraphics[width=\textwidth]{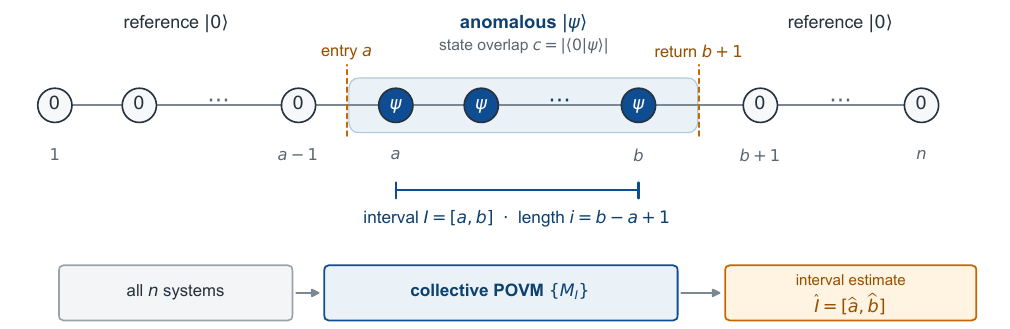}
  \caption{\textbf{Quantum change interval and collective localization.}
  A source emits the reference state \(\lvert0\rangle\) except on a nonempty
  contiguous interval \(I=[a,b]\), where it emits the known anomalous state
  \(\lvert\psi\rangle\). The entry and return boundaries are \(a\) and
  \(b+1\). For known length, \(i=b-a+1\) is fixed and only \(a\) is inferred;
  when the interval length is unknown, both endpoints are inferred. An
  arbitrary collective POVM may act on all \(n\) systems and return an interval estimate
  \(\widehat I\). Success requires \(\widehat I=I\). The local state overlap
  is \(c=\lvert\langle0\vert\psi\rangle\rvert\), so the states are
  nonorthogonal when \(0<c<1\).}
  \label{fig:model-geometry}
\end{figure*}
Success in this paper means exact recovery of the interval label:
\(\widehat I=I\). A minimum error positive operator valued measure (POVM)
reports an interval on every run and may be incorrect. By contrast, unambiguous exact identification allows for an inconclusive outcome but requires every
conclusive report to be correct \cite{sentis2017,llorens2024}.
Certified answer discrimination instead limits the distance between the
reported and true ordered labels \cite{martinez2019}. We study exact label
localization under the minimum error criterion and later add a hypothesis
\(H_0\) of no change.

For an unknown interval length, the Gram geometry is only locally
two-dimensional Toeplitz. Intersecting or adjacent intervals have a separable
kernel determined by endpoint displacement, while a positive gap introduces
an additional correction. The physical Gram matrix therefore lies beyond a
global two-dimensional Toeplitz compression or a transitive association
scheme, and its analysis requires more than a direct multilevel Szeg{\H o}
limit. Mohan, Sikora, and Upadhyay formulate general discrimination with two changes using
a finite Gram matrix and a Bayes SDP; their framework contains the
specialization in Table~\ref{tab:literature} \cite{mohan2023}. Here we derive
the fixed and moving overlap asymptotic laws for that returning interval
geometry and compare the explicit square root measurement (SRM) with the Bayes optimum.

\paragraph{Main results.}
Let \(p_1(c)\) be the one-dimensional Poisson kernel Toeplitz functional
defined in Eq.~\eqref{eq:p1-definition}, \(X\in\{\mathrm{tr},\mathrm{SRM},\mathrm{opt}\}\).

\begin{enumerate}
\def\labelenumi{\arabic{enumi}.}
\tightlist
\item
  Theorem~\ref{thm:fixed-length} gives the common Toeplitz limit for every fixed
  known length \(i\geq1\) and \(0\leq c\leq1\), together with an
  \(O_{i,c}(N^{-1/2})\) optimum--SRM gap for \(0<c<1\).
\item
  Theorem~\ref{thm:long-known} shows that the common limit is \(p_1(c^2)\)
  whenever the known length \(i_n\) and the number of translations
  \(N_n=n-i_n+1\) both diverge.
\item
  Theorem~\ref{thm:unknown-length} gives the common fixed overlap limit
  \(p_1(c)^2\) under the prior uniform over all
  \(M_n=n(n+1)/2\) intervals.
\item
  Theorem~\ref{thm:unknown-length-ultracritical} resolves the compact boundary
  window \(\tau_n=n(1-c_n)(\log n)^2=O(1)\), including a uniform logarithmic
  error estimate.
\item
  Theorem~\ref{thm:unknown-length-hull-dominance} extends the fixed overlap law to moving overlaps: \(P_X(G_n(c_n))\sim p_1(c_n)^2\) whenever \(c_n\to1\) from below and \(n p_1(c_n)\to\infty\).
\end{enumerate}

Theorem~\ref{thm:h0-localization} extends these results to include a
no change hypothesis. For a fixed prior \(\pi_0\), the conditional
localization limit \(\Lambda\) becomes
\(\pi_0+(1-\pi_0)\Lambda\). For fixed \(c<1\), the prior weighted SRM attains
this joint limit for both a growing known length and an unknown interval
length. At \(c=1\), maximum a posteriori postprocessing recovers the Bayes
optimum.

\paragraph{Proof strategy.}
At fixed overlap, Toeplitz comparison kernels are transferred to the physical
Gram matrices through local diagonal and perturbation estimates. For moving
overlaps, an exact weighted hull decomposition is combined with block, tail,
and endpoint macroblock estimates. The detailed Følner, transfer, Volterra,
triangular, and sector normalization arguments are organized in the
Supplemental Material. The two effective parameters have a geometric origin:
a unit translation of a long known length interval changes two local factors
and produces \(c^2\), whereas the two endpoint coordinates for unknown length
each carry a kernel with parameter \(c\). Full physical Gram SRM calculations
and small floating-point semidefinite programs provide finite size diagnostics.

\paragraph{Assumptions and scope.}
We consider conditionally independent pure state outputs with one nonempty
returning interval, known states, the stated uniform priors,
and unrestricted collective POVMs. The fixed overlap limits are pointwise
in \(c\), whereas
Section~\ref{a-nonuniform-boundary-at-unit-overlap} treats moving \(c_n\);
the no change prior \(\pi_0\) is fixed whenever \(H_0\) is included. These
results benchmark coherent storage of the full batch, followed by collective
readout. Extensions to mixed or correlated outputs, unknown anomaly states,
imperfect returns, restricted measurements, or multiple intervals require
separate analysis.

Table~\ref{tab:literature} situates the returning interval model relative to
permanent and multiple quantum changes, multi anomaly and edge detection
problems, and classical changed interval geometries. It records the exact
parameter maps and distinguishes the present collective POVM Bayes
localization task from unambiguous, online, and classical detection criteria.

\section{Model and preliminaries}\label{model-and-discrimination-preliminaries}

\subsection{Phase convention and Gram matrices}\label{phase-convention-and-gram-matrices}

Write

\[
\langle0|\psi\rangle=ce^{i\varphi},\qquad 0\le c\le1,
\qquad r=c^2.
\]

Replacing \(|\psi\rangle\) by \(e^{-i\varphi}|\psi\rangle\) changes each
candidate vector only by a global phase that depends on its label. The original
and phase fixed Gram matrices are therefore related by diagonal unitary
conjugation, which leaves the minimum error and SRM success probabilities
unchanged. We use

\[
\langle0|\psi\rangle=c\ge0
\]

without loss of generality.

Unless a theorem explicitly introduces a size-dependent model parameter,
such as \(i_n\) or \(c_n\), every asymptotic statement is taken as
\(n\to\infty\) with \(c\) (and hence \(r=c^2\)) fixed and with any stated
interval length \(i\) fixed. In
Section~\ref{joint-detection-and-exact-localization}, the no change prior
\(\pi_0\in(0,1)\) is also fixed rather than scaled with \(n\).

For an ensemble of \(M\) normalized pure states \(\{|\phi_j\rangle\}_{j=1}^M\)
under the uniform prior, let

\[
G_{jk}=\langle\phi_j|\phi_k\rangle.
\]

We write \(P_{\rm opt}(G)\) for the optimal minimum error success probability
over all POVMs
\cite{helstrom1976,holevo1973,yuen1975,barnett2009} and

\[
P_{\rm SRM}(G)=\frac1M\sum_{j=1}^M(\sqrt G)_{jj}^2
\]

for the square root measurement success probability
\cite{hausladen1994,eldar2001,dallapozza2015}. If \(G\) is singular,
\(G^{-1/2}\) is understood on the support of \(G\), and the POVM is completed
arbitrarily on its orthogonal complement; this does not change any ensemble
success probability. The SRM is an explicit collective POVM determined by the
candidate states and their average density operator. Evaluating it at a given
size therefore does not require solving the Bayes semidefinite program. The
elementary trace bound

\begin{equation}\label{eq:trace-srm-bound}
\left(\frac{\operatorname{tr}\sqrt G}{M}\right)^2
\le P_{\rm SRM}(G)\le P_{\rm opt}(G)
\end{equation}

will be used repeatedly.

\subsection{Supporting analytic estimates}\label{supporting-analytic-estimates}

The local diagonal Følner limit, operator norm Gram stability estimate, and
Gram dual bound used below are stated and proved in the Supplemental Material.
They transfer square root diagonal limits from positive convolution
compressions to the physical Gram ensembles, including singular Gram matrices.

\section{A known interval length}\label{a-known-interval-length}

Let the anomalous interval have known length \(i\ge1\). There are

\[
N=n-i+1
\]

admissible starting points, and the candidates under the uniform prior are

\begin{align*}
|\Phi_a^{(i)}\rangle
={}&|0\rangle^{\otimes(a-1)}|\psi\rangle^{\otimes i}
|0\rangle^{\otimes(n-a-i+1)}, \\
&a=1,\ldots,N.
\end{align*}

Their exact Gram matrix is

\begin{equation}\label{eq:fixed-gram}
(G_{N,i})_{ab}=r^{\min(|a-b|,i)}.
\end{equation}

Define

\begin{equation}\label{eq:fixed-symbol}
f_{i,r}(\theta)
=1-r^i+2\sum_{d=1}^{i-1}(r^d-r^i)\cos(d\theta).
\end{equation}

\subsection{Fixed length limit}\label{fixed-length-limit}

\begin{theorem}[Fixed length localization]\label{thm:fixed-length}
For every fixed integer \(i\ge1\) and \(0\le r\le1\),

\begin{equation}\label{eq:fixed-limit}
\lim_{N\to\infty}P_{\rm opt}(G_{N,i})
=\lim_{N\to\infty}P_{\rm SRM}(G_{N,i})
=P_\infty(i,r),
\end{equation}

where

\begin{equation}\label{eq:fixed-integral}
P_\infty(i,r)
=\left[
\frac1{2\pi}\int_0^{2\pi}\sqrt{f_{i,r}(\theta)}\,d\theta
\right]^2.
\end{equation}

For fixed \(i\) and \(0<c<1\),

\begin{equation}\label{eq:fixed-gap-rate}
0\le P_{\rm opt}(G_{N,i})-P_{\rm SRM}(G_{N,i})
=O_{i,c}(N^{-1/2}).
\end{equation}
\end{theorem}

\begin{proof}
\emph{Residual gap.}
Put \(q=r^i\) and decompose

\[
G_{N,i}=qJ_N+B_{N,i},
\]

where \(B_{N,i}=T_N(f_{i,r})\) has entries

\[
(B_{N,i})_{ab}=
\begin{cases}
r^{|a-b|}-r^i,&|a-b|<i,\\
0,&|a-b|\ge i.
\end{cases}
\]

Choose a local basis with
\(|\psi\rangle=\sqrt r\,|0\rangle+\sqrt{1-r}\,|1\rangle\). After removing the
common all-zero component, each candidate has a unique basis component with
squared amplitude \((1-r)^i\). Hence

\begin{equation}\label{eq:residual-spectral-gap}
B_{N,i}\ge(1-r)^iI_N.
\end{equation}

The corresponding infinite residual convolution operator has symbol
\(f_{i,r}\). Applying \eqref{eq:residual-spectral-gap} to arbitrarily long
finitely supported vectors gives

\begin{equation}\label{eq:symbol-spectral-gap}
f_{i,r}(\theta)\ge(1-r)^i.
\end{equation}

\emph{Circulant comparison.}
For \(N\ge2i\), let \(\widetilde B_{N,i}\) be the circulant matrix with the same local coefficients and set

\[
C_{N,i}=r^iJ_N+\widetilde B_{N,i}.
\]

Away from the uniform Fourier mode, its eigenvalues are the sampled values
\(f_{i,r}(2\pi k/N)\). The uniform mode receives an additional positive
rank-one contribution. Thus \(C_{N,i}\) is a positive definite, unit-diagonal
circulant Gram matrix. The matrices \(B_{N,i}\) and \(\widetilde B_{N,i}\)
differ only in the upper-right and lower-left wraparound corner blocks.

To make the boundary estimate explicit, set

\begin{align*}
b_d&=r^d-r^i,\\
E_{N,d}&=\sum_{a=1}^{d}e_a e_{N-d+a}^*,
\qquad 1\le d<i.
\end{align*}

Here \(E_{N,d}\) is a rank-\(d\) partial permutation from the last \(d\)
coordinates to the first \(d\) coordinates, and

\begin{equation}\label{eq:corner-partial-permutation}
\widetilde B_{N,i}-B_{N,i}
=\sum_{d=1}^{i-1}b_d(E_{N,d}+E_{N,d}^*).
\end{equation}

For \(N\ge2i\), the initial and terminal coordinate subspaces in
\eqref{eq:corner-partial-permutation} are disjoint. Consequently
\(\|E_{N,d}+E_{N,d}^*\|_1=2d\). Moreover, the whole difference maps between
the span of the first \(i-1\) coordinates and the span of the last \(i-1\)
coordinates. The triangle inequality and this two-corner support therefore give

\begin{equation}\label{eq:circulant-trace-perturbation}
\|G_{N,i}-C_{N,i}\|_1
\le K_{i,r}:=2\sum_{d=1}^{i-1}d(r^d-r^i),
\end{equation}

\begin{equation}\label{eq:circulant-rank-perturbation}
\operatorname{rank}(G_{N,i}-C_{N,i})\le2(i-1).
\end{equation}

The Sylvester integral gives

\begin{equation}\label{eq:sqrt-trace-perturbation}
\|\sqrt{G_{N,i}}-\sqrt{C_{N,i}}\|_1
\le\frac{K_{i,r}}{2(1-r)^{i/2}}=O_{i,r}(1).
\end{equation}

\emph{Exact circulant optimum.}
Because \(C_{N,i}\) and \(\sqrt{C_{N,i}}\) are circulant, the diagonal of \(\sqrt{C_{N,i}}\) is constant. If

\[
\beta_N=\frac1N\operatorname{tr}\sqrt{C_{N,i}},
\]

then the SRM value is \(\beta_N^2\). The Yuen--Kennedy--Lax dual operator

\[
Y=\frac{\beta_N}{N}\sqrt{C_{N,i}}
\]

is feasible because, with \(S=\sqrt{C_{N,i}}\),

\[
\beta_NS-Se_je_j^*S\ge0.
\]

Indeed, since \(S\succ0\), congruence by \(S^{-1/2}\) reduces this to

\[
\beta_N I-S^{1/2}e_je_j^*S^{1/2}\ge0.
\]

The rank-one term has the unique nonzero eigenvalue
\(e_j^*Se_j=\beta_N\), because the diagonal of \(S\) is constant. Since
\(\operatorname{tr}Y=\beta_N^2\), the SRM is exactly optimal for the
circulant model.

\emph{Transfer and spectral limit.}

Let

\[
\tau_N=\|\sqrt{G_{N,i}}-\sqrt{C_{N,i}}\|_1.
\]

Put \(\Delta_N=\sqrt{G_{N,i}}-\sqrt{C_{N,i}}\) and use the canonical Gram
realizations of the two ensembles. Cauchy--Schwarz, followed by monotonicity
of the Schatten norms, gives

\begin{align}\label{eq:opt-circulant-transfer}
|P_{\rm opt}(G_{N,i})-P_{\rm opt}(C_{N,i})|
&\le\frac1N\sum_{j=1}^N\|\Delta_Ne_j\|_2 \notag\\
&\le\frac1{\sqrt N}
\left(\sum_{j=1}^N\|\Delta_Ne_j\|_2^2\right)^{1/2} \notag\\
&=\frac{\|\Delta_N\|_{\rm HS}}{\sqrt N}
\notag\\
&\le\frac{\|\Delta_N\|_1}{\sqrt N}
=\frac{\tau_N}{\sqrt N}.
\end{align}

The SRM diagonal formula also gives

\[
|P_{\rm SRM}(G_{N,i})-P_{\rm SRM}(C_{N,i})|
\le\frac{2\tau_N}{N}.
\]

Together with \eqref{eq:sqrt-trace-perturbation}, this proves
\eqref{eq:fixed-gap-rate}. Finally, the sampled eigenvalues form a Riemann
sum. The rank-one correction to the uniform mode changes
\(N^{-1}\operatorname{tr}\sqrt C\) by at most \(O(N^{-1/2})\). Therefore

\[
\frac1N\operatorname{tr}\sqrt{C_{N,i}}
\longrightarrow
\frac1{2\pi}\int_0^{2\pi}\sqrt{f_{i,r}(\theta)}\,d\theta.
\]

This proves \eqref{eq:fixed-limit}--\eqref{eq:fixed-integral} for \(0<r<1\). The cases \(r=0,1\) are handled below.
\end{proof}

\begin{lemma}[Residual Toeplitz family]
Fix \(i\ge1\) and \(0\le r<1\), put \(q=r^i\), and let

\[
R_{N,i}:=\frac{G_{N,i}-qJ_N}{1-q}
=T_N\!\left(\frac{f_{i,r}}{1-r^i}\right).
\]

Then \(R_{N,i}\) is the Gram matrix of normalized residual states and

\begin{equation}\label{eq:residual-toeplitz-limit}
\lim_{N\to\infty}P_{\rm opt}(R_{N,i})
=\lim_{N\to\infty}P_{\rm SRM}(R_{N,i})
=\frac{P_\infty(i,r)}{1-r^i}.
\end{equation}

Equivalently,
\((1-r^i)P_{\rm opt}(R_{N,i})\to P_\infty(i,r)\), and the same statement
holds with \(P_{\rm SRM}\).
\end{lemma}

\begin{proof}
The diagonal of \(G_{N,i}-qJ_N\) is \(1-q\), and
\eqref{eq:residual-spectral-gap} shows that it is positive definite. Thus
\(R_{N,i}\) is a unit-diagonal Gram matrix. Its Toeplitz symbol is
\(h_{i,r}=f_{i,r}/(1-r^i)\), which is bounded above and below away from zero.
The one-dimensional local diagonal Følner limit, followed by the Gram dual
bound and \eqref{eq:trace-srm-bound}, gives the common
limit

\[
\left[\frac1{2\pi}\int_0^{2\pi}\sqrt{h_{i,r}(\theta)}\,d\theta\right]^2
=\frac{P_\infty(i,r)}{1-r^i}.
\]
\end{proof}

\subsection{Exact and special-function cases}\label{exact-and-special-function-cases}

\begin{corollary}[Intervals of one or two sites]
\label{cor:short-intervals}
\emph{One site.}
For \(i=1\),

\[
G_{N,1}=(1-r)I_N+rJ_N
\]

is permutation invariant. Its positive square root has constant diagonal, so
the preceding dual certificate proves finite size SRM optimality. Its
eigenvalues are \(1+(N-1)r\) once and \(1-r\) with multiplicity \(N-1\), so

\begin{align}\label{eq:one-site-finite}
P_{\rm opt}(G_{N,1})
&=P_{\rm SRM}(G_{N,1}) \notag\\
&=\frac{\bigl[\sqrt{1+(N-1)r}
+(N-1)\sqrt{1-r}\bigr]^2}{N^2}.
\end{align}

In particular, \(P_\infty(1,r)=1-r\).

This is the \(k=1\) specialization of the known single anomaly result in
Refs.~\cite{skotiniotis2024,llorens2024}, after identifying their common
off-diagonal Gram entry with \(r=c^2\). Thus the interval model reproduces the
established finite size formula at \(i=1\).

\par\medskip\noindent\emph{Two sites.}
With

\[
E(m)=\int_0^{\pi/2}\sqrt{1-m\sin^2u}\,du,
\]

one has

\begin{equation}\label{eq:two-site-elliptic}
P_\infty(2,r)
=\frac4{\pi^2}(1-r)(1+3r)
E^2\!\left(\frac{4r}{1+3r}\right).
\end{equation}
\end{corollary}

\begin{proof}[Derivation]
Equation~\eqref{eq:fixed-symbol} becomes

\begin{align*}
f_{2,r}(\theta)
&=(1-r)(1+r+2r\cos\theta) \\
&=(1-r)(1+3r-4r\sin^2(\theta/2)).
\end{align*}

With \(u=\theta/2\) and symmetry about \(\pi/2\),

\begin{align*}
&\frac1{2\pi}\int_0^{2\pi}\sqrt{f_{2,r}(\theta)}\,d\theta \\
&\qquad=\frac2\pi\sqrt{(1-r)(1+3r)}
E\!\left(\frac{4r}{1+3r}\right).
\end{align*}

Squaring proves \eqref{eq:two-site-elliptic}.
\end{proof}

\subsection{Endpoint parameters}\label{endpoint-parameters}

If \(r=0\), then \(G_{N,i}=I_N\), and both success probabilities equal one.
If \(r=1\), then \(G_{N,i}=J_N\), all hypotheses represent the same ray, and
both probabilities equal \(1/N\). These values agree with the continuous
limiting formulas. The spectral-gap arguments above apply only for
\(0<r<1\).

\section{A growing known length}\label{a-growing-known-length}

For \(0\le x<1\), define the Poisson kernel

\[
g_x(\theta)=\sum_{d\in\mathbb Z}x^{|d|}e^{id\theta}
=\frac{1-x^2}{1-2x\cos\theta+x^2}
\]

and

\begin{equation}\label{eq:p1-definition}
p_1(c)=\left[\frac1{2\pi}\int_0^{2\pi}\sqrt{g_x(\theta)}\,d\theta\right]^2
=\frac{4(1-x^2)}{\pi^2}K^2(x^2),
\end{equation}

where \(K\) is the complete elliptic integral of the first kind in the
parameter convention, and \(p_1(1)=0\) by continuity. The last equality follows
from the standard Landen transformation. We use the integral representation
throughout. The matrices with entries \(x^{|a-b|}\) form the classical
Kac--Murdock--Szeg{\H o} family \cite{kac1953,gray2006}.

\begin{lemma}[One change Toeplitz benchmark]
For

\[
Q_N=(r^{|a-b|})_{a,b=1}^N,
\qquad 0\le r<1,
\]

one has

\begin{align}\label{eq:one-change-limit}
P_{\rm opt}(Q_N),P_{\rm SRM}(Q_N)
&\longrightarrow p_1(r), \notag\\
\frac1N\operatorname{tr}\sqrt{Q_N}
&\longrightarrow\sqrt{p_1(r)}.
\end{align}
\end{lemma}

\begin{proof}
The matrix \(Q_N\) is the compression of the one-dimensional convolution
operator with kernel \(r^{|d|}\), whose symbol is \(g_r\) and is bounded above
and below away from zero. The local diagonal Følner limit, with \(d=1\) and
\(\phi(t)=\sqrt t\), gives convergence in mean of the square root diagonal to
\((2\pi)^{-1}\int\sqrt{g_r}=\sqrt{p_1(r)}\). The Gram dual bound
supplies the matching minimum error upper bound, while
\eqref{eq:trace-srm-bound} supplies the lower bound and the SRM squeeze. This
recovers the one-dimensional change-point Toeplitz functional of
Ref.~\cite{sentis2016}, evaluated here at the effective compound overlap
\(r=c^2\).
\end{proof}

\begin{theorem}[Growing known interval length]\label{thm:long-known}
Let \(i=i_n\) and \(N_n=n-i_n+1\). If

\begin{equation}\label{eq:long-known-conditions}
i_n\to\infty,\qquad N_n\to\infty,
\end{equation}

then, for every fixed \(0\le r\le1\),

\begin{equation}\label{eq:long-known-limit}
P_{\rm opt}(G_{N_n,i_n}),
P_{\rm SRM}(G_{N_n,i_n})
\longrightarrow p_1(r)=p_1(c^2).
\end{equation}

No condition on \(i_n/N_n\) is required.
The convergence is pointwise in fixed \(r\).
\end{theorem}

\begin{proof}
The endpoint values \(r=0,1\) follow from the matrices described in
Section~\ref{endpoint-parameters}. We therefore assume \(0<r<1\) and put
\(q_i=r^i\). Separating the common all zero component gives

\[
|\Phi_a^{(i)}\rangle
=\sqrt{q_i}|\Omega\rangle
+\sqrt{1-q_i}|\xi_a^{(i)}\rangle,
\qquad \langle\Omega|\xi_a^{(i)}\rangle=0.
\]

The normalized residual Gram matrix is

\begin{equation}\label{eq:residual-gram}
(R_{N,i})_{ab}=
\begin{cases}
\dfrac{r^{|a-b|}-r^i}{1-r^i},&|a-b|<i,\\
0,&|a-b|\ge i.
\end{cases}
\end{equation}

Its symbol

\[
h_{i,r}(\theta)
=1+2\sum_{d=1}^{i-1}\frac{r^d-r^i}{1-r^i}\cos(d\theta)
\]

satisfies

\begin{equation}\label{eq:residual-symbol-convergence}
\|h_{i,r}-g_r\|_\infty
\le\frac{2(i-1)r^i}{1-r^i}+\frac{2r^i}{1-r}
=:\varepsilon_i(r)\longrightarrow0.
\end{equation}

Let \(Q_N=(r^{|a-b|})_{a,b=1}^N\). Toeplitz compression gives

\begin{equation}\label{eq:residual-operator-convergence}
\|R_{N,i}-Q_N\|_{\rm op}\le\varepsilon_i(r).
\end{equation}

Since \(g_r\ge(1-r)/(1+r)>0\), both matrices have a common positive
spectral gap for large \(i\). The operator norm Gram stability estimate therefore
transfers the optimal minimum error success probability uniformly in \(N\).
The trace distance
between each original state and its residual state is also
\(\sqrt{q_i}=r^{i/2}\), so

\[
|P_{\rm opt}(G_{N,i})-P_{\rm opt}(R_{N,i})|
\le r^{i/2}.
\]

The single-change-point Toeplitz limit for \(Q_N\) therefore proves

\[
P_{\rm opt}(G_{N,i})\to p_1(r).
\]

For the SRM, the rank-one common component requires a separate trace estimate.
Since

\[
G_{N,i}=q_iJ_N+(1-q_i)R_{N,i},
\]

operator monotonicity and the Rotfel'd trace inequality give the exact sandwich

\begin{align}\label{eq:rank-one-sandwich}
\sqrt{1-q_i}\,\operatorname{tr}\sqrt{R_{N,i}}
&\le\operatorname{tr}\sqrt{G_{N,i}} \notag\\
&\le\sqrt{1-q_i}\,\operatorname{tr}\sqrt{R_{N,i}}
+\sqrt{q_iN}.
\end{align}

Equation~\eqref{eq:residual-operator-convergence}, the common spectral gap, and the Toeplitz limit imply

\[
\frac1N\operatorname{tr}\sqrt{R_{N,i}}
\longrightarrow\sqrt{p_1(r)}.
\]

After dividing \eqref{eq:rank-one-sandwich} by \(N\), the remaining rank-one term is \(\sqrt{q_i/N}\to0\). Thus

\[
\frac1N\operatorname{tr}\sqrt{G_{N,i}}
\longrightarrow\sqrt{p_1(r)}.
\]

The sandwich \eqref{eq:trace-srm-bound} completes the SRM proof.
\end{proof}
Figure~\ref{fig:analytic-limits}(a) displays the family of
fixed length limits $P_\infty(i,c^2)$; as $i$ grows these limits
approach the long-interval law $p_1(c^2)$ derived in
Theorem~\ref{thm:long-known}.
\begin{corollary}[An exact finite size regime]
If \(i\ge N-1\), then \(\min(|a-b|,i)=|a-b|\) for every matrix entry, and hence

\[
G_{N,i}=Q_N
\]

exactly. The approximation in Theorem~\ref{thm:long-known} is needed only when \(i<N-1\).
\end{corollary}

\begin{figure*}[t]
  \centering
  \includegraphics[width=0.96\textwidth]{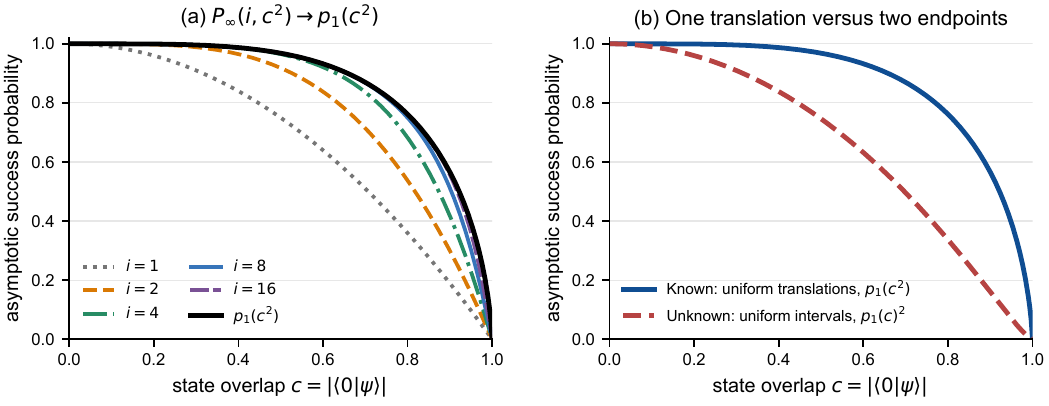}
  \caption{\textbf{Hierarchy of analytic localization limits.}
  (a) The fixed length Toeplitz integral \(P_\infty(i,c^2)\) approaches the
  law \(p_1(c^2)\) for a growing known interval length. The \(i=2\) curve is
  equivalently given by the complete elliptic integral expression in
  Corollary~\ref{cor:short-intervals}. (b) The one translation law \(p_1(c^2)\) and
  two-endpoint law \(p_1(c)^2\) are shown for their respective uniform
  translation and uniform interval priors. The latter follows from the
  two-dimensional Følner limit and separable comparison symbol.}
  \label{fig:analytic-limits}
\end{figure*}

\section{Unknown interval length}\label{unknown-interval-length}

Let

\begin{align*}
\mathcal I_n&=\{[a,b]:1\le a\le b\le n\}, \\
M_n&=|\mathcal I_n|=\frac{n(n+1)}2,
\end{align*}

and equip \(\mathcal I_n\) with the uniform prior \(1/M_n\). This is uniform
over intervals, not over lengths. If \(L=|I|\), the induced marginal is

\begin{equation}\label{eq:induced-length-prior}
\begin{aligned}
\Pr\{L=\ell\}
&=\frac{n-\ell+1}{M_n}\\
&=\frac{2(n-\ell+1)}{n(n+1)},
\qquad 1\leq\ell\leq n.
\end{aligned}
\end{equation}

Conditionally on \(L=\ell\), the starting point is uniform over its
\(n-\ell+1\) admissible values. Consequently
\(\mathbb E L=(n+2)/3\), and \(L/n\) converges to the triangular density
\(2(1-x)\) on \(0<x<1\). Thus the prior favors shorter lengths, but its
typical interval still has order-\(n\) length. The candidate indexed by
\(I=[a,b]\) has \(|\psi\rangle\) on \(I\) and \(|0\rangle\) elsewhere. Its
exact Gram kernel is

\begin{equation}\label{eq:interval-gram}
G^{(n)}_{I,J}=c^{|I\triangle J|}.
\end{equation}

Theorem~\ref{thm:unknown-length} is therefore a Bayes law under the uniform
interval prior in \eqref{eq:induced-length-prior}. Other priors may retain
nonvanishing mass on bounded or logarithmically short intervals and can lead
to different limiting reductions.

For intersecting intervals,

\begin{equation}\label{eq:intersecting-distance}
|I\triangle J|=|a-a'|+|b-b'|.
\end{equation}

For disjoint intervals, assume without loss of generality that \(b<a'\) and define the separating gap by

\[
g(I,J):=a'-b-1\ge0.
\]

Then the exact identity is

\begin{equation}\label{eq:disjoint-distance}
|a-a'|+|b-b'|=|I\triangle J|+2g(I,J).
\end{equation}

Thus \eqref{eq:intersecting-distance} still holds for adjacent intervals,
where \(g(I,J)=0\), but differs by twice the separating gap when
\(g(I,J)>0\). The exact physical kernel is therefore not globally a
two-dimensional Toeplitz kernel.

Classical multilevel Szeg{\H o} theorems describe spectral averages for genuine
multilevel Toeplitz compressions \cite{tilli1998,tyrtyshnikov1998}. Here the
triangular endpoint sets and the non-Toeplitz physical kernel require a local
diagonal Følner statement and a separate perturbative transfer back to
\(G^{(n)}\).

\paragraph{Transfer from the comparison kernel.}
The retained family Gram transfer and its Følner specialization are stated and
proved in the Supplemental Material. Applied to a retained family with
asymptotically full prior mass whose Gram matrix converges in operator norm to
a positive convolution compression, these results transfer the square root
diagonal limit of the comparison matrix to both the SRM and the Bayes optimum.

\subsection{The two-endpoint law}\label{the-two-endpoint-law}

\begin{theorem}[Unknown interval length]\label{thm:unknown-length}
For every fixed \(0\le c\le1\), under the uniform prior on \(\mathcal I_n\),

\begin{equation}\label{eq:unknown-length-limit}
P_{\rm opt}(G^{(n)}),P_{\rm SRM}(G^{(n)})
\longrightarrow p_1(c)^2
\end{equation}

This fixed overlap statement is pointwise in \(c\); moving overlaps are treated
in Section~\ref{a-nonuniform-boundary-at-unit-overlap}.
\end{theorem}

\begin{proof}
\emph{Retained family.}
For \(c=0\), the hypotheses are orthogonal, whereas for \(c=1\) they are
identical; both cases agree with the stated limit. Assume henceforth that
\(0<c<1\), and choose

\[
L_n=\lceil\gamma\log n\rceil,
\qquad \gamma>|\log c|^{-1},
\]

and retain only

\[
\mathcal B_n=\{I\in\mathcal I_n:|I|\ge L_n\}.
\]

The deleted prior mass is

\begin{equation}\label{eq:short-interval-mass}
\varepsilon_n
=1-\frac{|\mathcal B_n|}{M_n}
=O\!\left(\frac{\log n}{n}\right)\to0.
\end{equation}

On \(\mathcal B_n\), define the comparison kernel

\begin{equation}\label{eq:comparison-kernel}
H^{(n)}_{I,J}=c^{|a-a'|+|b-b'|}.
\end{equation}

It is the compression of the positive convolution kernel \(c^{|u|+|v|}\). The two kernels agree for intersecting intervals. For disjoint intervals,

\[
|G_{I,J}-H_{I,J}|
\le c^{|I|+|J|}\le c^{2L_n}.
\]

Consequently,

\begin{equation}\label{eq:comparison-operator}
\|G_{\mathcal B_n}-H_{\mathcal B_n}\|_{\rm op}
\le|\mathcal B_n|c^{2L_n}=o(1).
\end{equation}

\emph{Følner geometry.}
The retained endpoint sets form a Følner sequence. Under the endpoint
identification \(I=[a,b]\leftrightarrow(a,b)\),

\[
\mathcal B_n
=\{(a,b)\in\mathbb Z^2:a\ge1,\ b\le n,\ b-a\ge L_n-1\},
\]

and

\begin{equation}\label{eq:long-domain-cardinality}
|\mathcal B_n|
=\frac{(n-L_n+1)(n-L_n+2)}2
\sim\frac{n^2}{2}.
\end{equation}

Fix \(h=(u,v)\in\mathbb Z^2\) and put \(H=\|h\|_\infty\). If the indicators
of \(\mathcal B_n\) at \(x=(a,b)\) and \(x+h\) differ, at least one of the
three affine inequalities defining \(\mathcal B_n\) changes truth value.
Such an \(x\) lies in a strip of respectively \(|u|\), \(|v|\), or
\(|v-u|\) lattice layers adjacent to
\(a=1\), \(b=n\), or \(b-a=L_n-1\). Because either \(x\) or \(x+h\) belongs
to \(\mathcal B_n\), each layer contains at most \(n+2H+1\) relevant points.
Put \(C_h=|u|+|v|+|v-u|\). Therefore

\begin{equation}\label{eq:long-domain-folner-bound}
\frac{|(\mathcal B_n-h)\triangle\mathcal B_n|}{|\mathcal B_n|}
\le
\frac{C_h(n+2H+1)}{|\mathcal B_n|}
\longrightarrow0,
\end{equation}

where \eqref{eq:long-domain-cardinality} and \(L_n=O(\log n)=o(n)\) were
used. Thus \((\mathcal B_n)\) is a Følner sequence in \(\mathbb Z^2\).
The comparison convolution symbol is

\[
F_c(\theta_1,\theta_2)=g_c(\theta_1)g_c(\theta_2),
\]

with a uniform positive lower bound. The square root spectral average in
the Følner comparison transfer factorizes as

\begin{align}\label{eq:endpoint-factorization}
\alpha_c
&=\frac1{(2\pi)^2}\int\!\!\int\sqrt{F_c} \notag\\
&=\left[\frac1{2\pi}\int_0^{2\pi}
\sqrt{g_c(\theta)}\,d\theta\right]^2
=p_1(c).
\end{align}

\emph{Transfer to the full ensemble.}
The kernel \(k(u,v)=c^{|u|+|v|}\) is Hermitian, belongs to
\(\ell^1(\mathbb Z^2)\), and has \(k(0,0)=1\). Equations
\eqref{eq:short-interval-mass}, \eqref{eq:comparison-operator}, and
\eqref{eq:long-domain-folner-bound} verify all remaining hypotheses of
the Følner comparison transfer. That result transfers the
comparison-kernel limit simultaneously to the optimum and the SRM of the full
physical Gram ensemble, giving
\(\alpha_c^2=p_1(c)^2\).
\end{proof}
Figure~\ref{fig:analytic-limits}(b) contrasts the two laws:
$p_1(c^2)$ for uniform translations of a known long interval and
$p_1(c)^2$ for the uniform prior over intervals of unknown length.

\subsection{The boundary layer at unit overlap}
\label{a-nonuniform-boundary-at-unit-overlap}

The pointwise fixed overlap limit in Theorem~\ref{thm:unknown-length} does not
resolve sequences with \(c_n\to1\). At \(c_n=1\), all interval labels
represent the same state. The trace benchmark, SRM success probability, and
Bayes optimum then all equal the exact label baseline \(1/M_n\). For a moving
overlap, write

\begin{equation}\label{eq:moving-physical-gram}
 G_n(c)=\bigl(c^{|I\triangle J|}\bigr)_{I,J\in\mathcal I_n},
 \qquad G_n:=G_n(c_n),
\end{equation}

We use the following boundary scale for the amplitude overlap:

\begin{equation}\label{eq:critical-parameter}
 L_n=\log n,
 \qquad
 \tau_n=n(1-c_n)L_n^2.
\end{equation}

The corresponding squared-overlap quantity is \((1+c_n)\tau_n\). Put

\begin{equation}\label{eq:critical-amplitude}
 A(t)=1+\frac{2\sqrt t}{\pi}+\frac{\sqrt2\,t}{\pi^2},
 \qquad F(t)=A(t)^2,
\end{equation}

and, for an \(M\)-label Gram matrix, define the trace benchmark

\begin{equation}\label{eq:critical-trace-benchmark}
 P_{\rm tr}(G)=\left(\frac{\operatorname{tr}\sqrt G}{M}\right)^2.
\end{equation}

For \(0\le t\le T\), set

\begin{equation}\label{eq:critical-overlap-family}
 c_n(t)=1-\frac{t}{nL_n^2},
 \qquad
 (G_{n,t})_{I,J}=c_n(t)^{|I\triangle J|}.
\end{equation}

For all sufficiently large \(n\), \(c_n(t)>0\) uniformly on every fixed
interval \([0,T]\). At \(t=0\), \(G_{n,0}=J_{M_n}\), and the three success
criteria are each exactly \(1/M_n\).

In this subsection, a maximum indexed by \(X\) ranges over
\(\{\mathrm{tr},\mathrm{SRM},\mathrm{opt}\}\).

\begin{theorem}[Ultracritical law for unknown interval length]
\label{thm:unknown-length-ultracritical}
For every \(T\in[0,\infty)\), there are \(C_T<\infty\) and \(n_T\) such that,
for \(n\ge n_T\),

\begin{equation}\label{eq:critical-uniform-rate}
 \sup_{0\le t\le T}
 \max_X
 \left|M_nP_X(G_{n,t})-F(t)\right|
 \le C_T\frac{\log\log n}{\log n}.
\end{equation}

Consequently, if \(\tau_n\to\tau\in[0,\infty)\) in
\eqref{eq:critical-parameter}, then

\begin{equation}\label{eq:critical-common-limit}
 M_nP_X(G_n)\longrightarrow F(\tau)
 \quad\text{for every }X.
\end{equation}

More precisely, for the same \(n\ge n_T\), if
\(\tau,\tau_n\in[0,T]\) and \(G_n=G_{n,\tau_n}\), then, after increasing
\(C_T\) if necessary,

\begin{equation}\label{eq:critical-centered-rate}
 \begin{aligned}
 &\max_X\left|M_nP_X(G_n)-F(\tau)\right|\\[-0.3ex]
 &\quad\le C_T\left[
 \frac{\log\log n}{\log n}
 +\sqrt{|\tau_n-\tau|}
 +|\tau_n-\tau|
 \right].
 \end{aligned}
\end{equation}
\end{theorem}

The moving overlap problem has two asymptotic regimes. In the compact
critical window \(\tau_n=O(1)\), the vacuum and one excitation
contributions remain comparable to the interval hull contribution. When
\(n p_1(c_n)\to\infty\), the interval hull dominates and
\[
P_X(G_n(c_n))\sim p_1(c_n)^2.
\]
These two regimes are established in
Theorems~\ref{thm:unknown-length-ultracritical}
and~\ref{thm:unknown-length-hull-dominance}, respectively.

\paragraph{Proof idea.}
For \(0<c_n\leq1\), set
\[
q_n=\frac{\sqrt{1-c_n^2}}{c_n},
\qquad
(D_n)_{I,I}=c_n^{|I|},
\]
and decompose the interval states into the vacuum sector, the sector with one
excitation, and the sectors with two or more excitations. The two nonvacuum
parts have the exact factorizations
\begin{equation}\label{eq:critical-sector-factorization}
A_{1,n}=q_nD_nB_{1,n},
\qquad
A_{\geq2,n}=q_n^2D_nZ_nW_n,
\end{equation}
where \(B_{1,n}\) is the interval-site incidence matrix, \(Z_n\) is the
proper-hull incidence matrix, and \(W_n\) contains the physical hull weights.

In the critical window, the vacuum, one excitation, and higher excitation
sectors contribute the leading amplitudes
\[
1,\qquad
\frac{2\sqrt t}{\pi},\qquad
\frac{\sqrt2\,t}{\pi^2},
\]
whose sum is \(A(t)\). A global square root trace estimate gives the lower
bound, while sectorwise Bayes dual bounds give the matching upper bound.
Together with
\(P_{\rm tr}\leq P_{\rm SRM}\leq P_{\rm opt}\), these estimates give
\[
A(t)-\varepsilon_n
\leq
\sqrt{M_nP_X(G_{n,t})}
\leq
A(t)+\varepsilon_n,
\qquad
\varepsilon_n
=
O_T\!\left(\frac{\log L_n}{L_n}\right),
\]
uniformly for \(0\leq t\leq T\) and
\(X\in\{\mathrm{tr},\mathrm{SRM},\mathrm{opt}\}\). Since
\(L_n=\log n\), squaring gives
\eqref{eq:critical-uniform-rate}. Continuity of \(F=A^2\) gives the remaining
statements. The incidence trace estimates and the Bayes dual constructions
are proved in the Supplemental Material.

\begin{lemma}[Inner matching to the fixed overlap law]
Suppose

\begin{equation}\label{eq:critical-inner-window}
 \tau_n=n(1-c_n)L_n^2\longrightarrow\infty,
 \qquad n(1-c_n)=\frac{\tau_n}{L_n^2}\longrightarrow0.
\end{equation}

Then, for \(X\in\{\mathrm{tr},\mathrm{SRM},\mathrm{opt}\}\),

\begin{equation}\label{eq:critical-inner-law}
 M_nP_X(G_n)\sim\frac{2\tau_n^2}{\pi^4}\sim A(\tau_n)^2.
\end{equation}

Moreover, this leading term agrees with the endpoint expansion of the
fixed overlap law:

\begin{equation}\label{eq:critical-inner-p1-matching}
 M_np_1(c_n)^2\sim\frac{2\tau_n^2}{\pi^4}.
\end{equation}
\end{lemma}

\paragraph{Proof idea.}
Under \eqref{eq:critical-inner-window}, the physical hull weights converge
uniformly to one. The triangular Volterra trace law and its matching Bayes
dual bound then give the leading amplitude
\(\sqrt2\,\tau_n/\pi^2\) in the sector with two or more excitations. The
vacuum and one excitation amplitudes are only \(O(1)\) and
\(O(\sqrt{\tau_n})\), so they are lower order as \(\tau_n\to\infty\).
The sector trace and dual bounds therefore give
\eqref{eq:critical-inner-law}.

The endpoint expansion
\[
p_1(c)^2
=
\frac{4(1-c)^2}{\pi^4}
\log^4\!\left(\frac{8}{1-c}\right)
[1+o(1)]
\]
and the identity
\(1-c_n=\tau_n/(nL_n^2)\) then give
\eqref{eq:critical-inner-p1-matching}. The uniform estimates needed for this
joint limit are proved in the Supplemental Material.

\begin{theorem}[Hull dominance criterion]
\label{thm:unknown-length-hull-dominance}
Suppose

\begin{equation}\label{eq:hull-dominance-assumptions}
 \begin{aligned}
 c_n&\to1,\qquad c_n<1\ \text{eventually},\\
 h_n&:=n p_1(c_n)\to\infty.
 \end{aligned}
\end{equation}

No monotonicity condition on \(c_n\) is imposed. Simultaneously for
\(X\in\{\mathrm{tr},\mathrm{SRM},\mathrm{opt}\}\), under the uniform prior
\(1/M_n\) on all nonempty intervals,

\begin{equation}\label{eq:hull-dominance-law}
 \frac{P_X(G_n(c_n))}{p_1(c_n)^2}\longrightarrow1.
\end{equation}
\end{theorem}

\paragraph{Proof idea.}
Set
\[
\lambda_n=n(1-c_n).
\]
Every subsequence has a further subsequence on which \(\lambda_n\) tends to
zero, to a finite positive limit, or to infinity. If
\(\lambda_n\to0\), the endpoint expansion of \(p_1\), together with
\(n p_1(c_n)\to\infty\), forces \(\tau_n\to\infty\), and the inner matching
result applies. A uniform continuum estimate covers the case
\(\lambda_n\to\lambda\in(0,\infty)\), while the moving outer estimate covers
\(\lambda_n\to\infty\). Each regime gives
\eqref{eq:hull-dominance-law}. The subsequence criterion then proves the full
limit without a monotonicity assumption on \(c_n\). The uniform estimates for
the three regimes are established in the Supplemental Material.

The compact continuum interface proposition also gives the stronger uniform
statement that, for every compact set \(K\subset(0,\infty)\),

\begin{equation}\label{eq:continuum-interface-law}
 \sup_{\lambda\in K}\max_X
 \left|
 \frac{M_nP_X\bigl(G_n(1-\lambda/n)\bigr)}{(\log n)^4}
 -\frac{2\lambda^2}{\pi^4}
 \right|\longrightarrow0.
\end{equation}

The hull dominance criterion also includes the moving outer regime
\(n(1-c_n)\to\infty\). Conditions
\eqref{eq:hull-dominance-assumptions} are sufficient; whether they are
necessary remains open.

\section{No change extension: joint detection and exact localization}
\label{joint-detection-and-exact-localization}

We now allow the reference state

\[
 |\Omega\rangle=|0\rangle^{\otimes n}
\]

as a no change hypothesis \(H_0\) with fixed prior
\(\pi_0\in(0,1)\). The overlap parameter and \(\pi_0\) do not scale with
\(n\). An explicit no change label also appears in the at-most-one- and
at-most-two-change formulations of Ref.~\cite{mohan2023}; here the alternative
is one nonempty returning interval.

For known length, the complete joint prior is

\begin{equation}\label{eq:h0-known-joint-prior}
\Pr(H_0)=\pi_0,
\qquad
\Pr(a)=\frac{1-\pi_0}{N},
\quad 1\le a\le N.
\end{equation}

When the interval length is unknown, the joint prior is

\begin{equation}\label{eq:h0-unknown-joint-prior}
\Pr(H_0)=\pi_0,
\qquad
\Pr(I)=\frac{1-\pi_0}{M_n},
\quad I\in\mathcal I_n.
\end{equation}

Thus the anomaly label is conditionally uniform. Success requires reporting
\(H_0\) for \(|\Omega\rangle\) and the exact interval otherwise. We write
\(P_{\rm joint}^*\) for the Bayes optimum over all POVMs. For a pure-state
ensemble with priors \(p_j\), let

\begin{equation}\label{eq:weighted-gram-srm}
W_{jk}=\sqrt{p_jp_k}\langle\phi_j|\phi_k\rangle,
\qquad
\mathcal P(W)=\sum_j\left|(\sqrt W)_{jj}\right|^2.
\end{equation}

This is the success probability of the prior weighted square root measurement.
For a uniform prior on \(M\) labels, \(W=G/M\) and
\(\mathcal P(W)=P_{\rm SRM}(G)\). We denote the success probability for the
ensemble augmented by \(H_0\) by \(P_{\rm wSRM}^{H_0}\).

\begin{theorem}[Joint localization with a no change hypothesis]
\label{thm:h0-localization}
Fix \(\pi_0\in(0,1)\).

\emph{Fixed known length.}
Under the joint prior \eqref{eq:h0-known-joint-prior}, for fixed \(i\) and
\(0\le r\le1\), as \(N\to\infty\),

\begin{equation}\label{eq:h0-fixed-limit}
P_{\rm joint}^*(N,i,r,\pi_0)
\longrightarrow
\pi_0+(1-\pi_0)P_\infty(i,r).
\end{equation}

\emph{Growing known length.}
For fixed \(0\le c\le1\), if \(i_n\to\infty\) and \(N_n\to\infty\),
then

\begin{equation}\label{eq:h0-growing-limit}
P_{\rm joint}^*
\longrightarrow
\pi_0+(1-\pi_0)p_1(c^2).
\end{equation}

\emph{Unknown length.}
For fixed \(0\le c\le1\), under the joint prior
\eqref{eq:h0-unknown-joint-prior}, as \(n\to\infty\),

\begin{equation}\label{eq:h0-unknown-limit}
P_{\rm joint}^*(n,c,\pi_0)
\longrightarrow
\pi_0+(1-\pi_0)p_1(c)^2.
\end{equation}

Moreover, for fixed \(0\le c<1\), the prior weighted square root measurement
attains the joint Bayes limit in the two diverging-length models:

\begin{equation}\label{eq:h0-growing-weighted-srm-limit}
\begin{aligned}
P_{\rm wSRM}^{H_0}
&\longrightarrow\pi_0+(1-\pi_0)p_1(c^2),\\
P_{\rm joint}^*-P_{\rm wSRM}^{H_0}
&\longrightarrow0
\qquad(i_n,N_n\to\infty),
\end{aligned}
\end{equation}

and

\begin{equation}\label{eq:h0-unknown-weighted-srm-limit}
\begin{aligned}
P_{\rm wSRM}^{H_0}
&\longrightarrow\pi_0+(1-\pi_0)p_1(c)^2,\\
P_{\rm joint}^*-P_{\rm wSRM}^{H_0}
&\longrightarrow0
\qquad(n\to\infty).
\end{aligned}
\end{equation}
\end{theorem}

\paragraph{Proof idea.}
For any joint POVM, the contribution from \(H_0\) is at most \(\pi_0\), and
the remaining measurement operators can be completed to a POVM for the
conditional interval ensemble. This gives the upper bound. For the reverse
bound, we reserve \(\lvert\Omega\rangle\) for \(H_0\) and apply an optimal
localization measurement to the normalized residual states in
\(\lvert\Omega\rangle^\perp\). The residual Toeplitz limit gives the fixed
known length result. For unknown length, the same construction is applied
after removing intervals shorter than \(m_n\asymp\log n\), whose total prior
mass vanishes. These bounds give the three affine Bayes limits.

For the prior weighted SRM, the only coupling between the \(H_0\) and anomaly
blocks of the weighted Gram matrix is the vector of reference overlaps. The
corresponding mean squared overlap is \(c^{2i_n}\) for a growing known length
and \(O_c(n^{-1})\) for unknown length. Both vanish for fixed \(c<1\), so
stability of the positive square root transfers the conditional SRM limits to
the joint ensemble. Complete proofs and quantitative bounds are given in the
Supplemental Material.

For a fixed known length, the mean squared reference overlap is \(c^{2i}\)
and does not vanish with \(N\). The same decoupling argument therefore does
not apply to the prior-weighted SRM in this regime.

\begin{proposition}[Unit overlap and MAP postprocessing]
\label{prop:h0-unit-map}
At \(c=1\), let \(K=N\) for known length and \(K=M_n\) for unknown length.
Then, exactly at finite size,

\begin{equation}\label{eq:h0-weighted-srm-unit-overlap}
\begin{aligned}
P_{\rm wSRM}^{H_0}
&=\pi_0^2+\frac{(1-\pi_0)^2}{K},\\
P_{\rm joint}^*
&=\max\!\left\{\pi_0,\frac{1-\pi_0}{K}\right\}.
\end{aligned}
\end{equation}

Consequently, as \(K\to\infty\),

\begin{equation}\label{eq:h0-weighted-srm-unit-gap}
\begin{aligned}
P_{\rm wSRM}^{H_0}&\to\pi_0^2,\\
P_{\rm joint}^*-P_{\rm wSRM}^{H_0}
&\to\pi_0(1-\pi_0)>0.
\end{aligned}
\end{equation}

For a POVM \(\{M_y\}_y\), maximum a posteriori (MAP) postprocessing assigns
outcome \(y\) to a label in

\[
 \mathop{\rm arg\,max}_j
 p_j\operatorname{tr}(\rho_jM_y),
\]

with an arbitrary fixed rule for ties. Applied after the prior weighted
square root measurement, this rule is asymptotically Bayes optimal for every
fixed \(0\le c\le1\), both when \(i_n,N_n\to\infty\) and when the interval
length is unknown and \(n\to\infty\).
\end{proposition}

\paragraph{Proof idea.}
At \(c=1\), all hypotheses represent the same physical state. If \(p_j\) are
their priors and \(u_j=\sqrt{p_j}\), then the weighted Gram matrix is the rank
one projection \(W=uu^*\). Hence \(\sqrt W=W\), the prior-weighted SRM has
success probability \(\sum_jp_j^2\), and the Bayes rule selects the label with
the largest prior. Substituting \(p_0=\pi_0\) and
\(p_j=(1-\pi_0)/K\) gives
Eqs.~\eqref{eq:h0-weighted-srm-unit-overlap}
and~\eqref{eq:h0-weighted-srm-unit-gap}.

For any POVM, MAP postprocessing maximizes the contribution of each
measurement outcome and therefore cannot reduce the total success
probability. For fixed \(c<1\), Theorem~\ref{thm:h0-localization} and the
Bayes upper bound squeeze the postprocessed success probability to the joint
optimum. At \(c=1\), the outcome distribution is independent of the true
label. For all sufficiently large \(K\), the MAP rule therefore reports
\(H_0\) and attains the Bayes optimum \(\pi_0\). The complete argument is
given in the Supplemental Material.

\section{Numerical results}
\label{numerical-verification}

The numerical calculations illustrate finite size behavior for the physical
Gram ensembles. Figure~\ref{fig:finite size} summarizes the results.
Panels~(a) and (b) compare SRM success with \(p_1(c^2)\) and \(p_1(c)^2\)
for a growing known length and unknown interval length, respectively. Panels~(b) shows stronger finite size effects at larger
overlap. The correlation length

\[
\xi(c)=|\log c|^{-1}
\]
is approximately \(99.5\) at \(c=0.99\), so even \(n=50\) has
\(n/\xi\approx0.503\). The finite size deviations at high overlap are
consistent with a slow approach to the asymptotic regime.
Panel~(c) displays \(U_{\rm safe}-P_{\rm SRM}\) over the complete grid of
small SDP instances computed in canonical Gram coordinates.
The numerical grid covers \(n=3,\ldots,7\) and

\[
c\in\{0.3,0.6,0.8,0.9,0.95,0.99\}.
\]

\begin{figure*}[t]
  \centering
  \includegraphics[width=\textwidth]{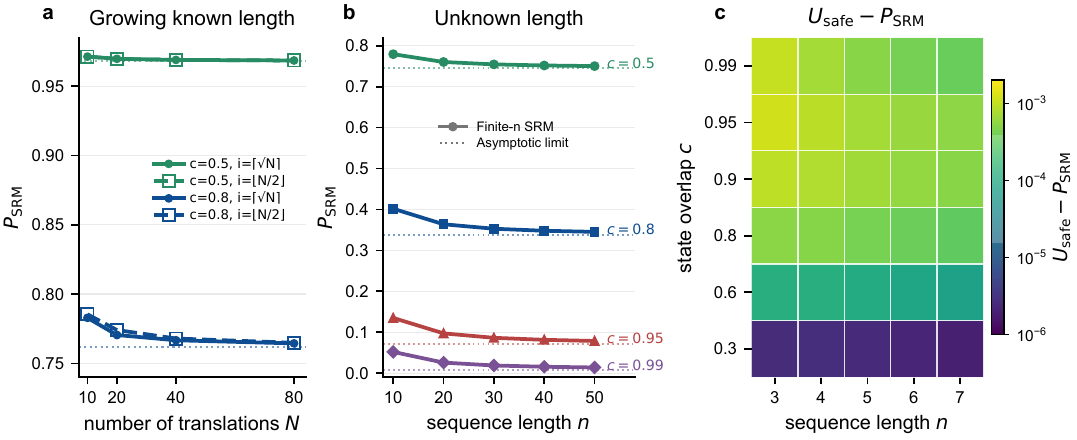}
  \caption{\textbf{Finite size SRM behavior and SDP diagnostics for small
  instances.}
  (a) SRM success for two schedules with a growing
  known length,
  \(i=\lceil\sqrt N\rceil\) and \(i=\lfloor N/2\rfloor\) at
  \(c=0.5,0.8\); dotted horizontal lines are \(p_1(c^2)\). Open squares are
  drawn around filled circles where the two schedules are nearly coincident.
  (b) Dense SRM computations using the full physical Gram matrix through
  \(n=50\) for a prior
  uniform over all nonempty intervals; dotted lines are \(p_1(c)^2\). The
  \(c=0.95,0.99\) curves show the stronger finite size effects at large
  overlap.
  (c) Logarithmic color map of the floating point diagnostic
  \(U_{\rm safe}-P_{\rm SRM}\) for the complete grid of 30 semidefinite
  programs, with
  \(n=3,\ldots,7\) and
  \(c\in\{0.3,0.6,0.8,0.9,0.95,0.99\}\).}
  \label{fig:finite size}
\end{figure*}

The postprocessed feasibility checks give

\[
 L_{\rm safe}\le P_{\rm opt}\le U_{\rm safe}.
\]

Because the implementation uses IEEE double precision, panel~(c) is a floating
point diagnostic rather than a rigorous numerical enclosure; the complete
postprocessing construction and roundoff limitations are given in the
Supplemental Material. Across the 30 cases with unknown interval length, the
largest raw absolute primal--dual gap is \(6.782\times10^{-9}\), while the
largest postprocessed floating point bracket width
\(U_{\rm safe}-L_{\rm safe}\) is \(1.025\times10^{-8}\). The largest rechecked
completeness residual is \(1.41\times10^{-15}\). The smallest postprocessed
primal eigenvalue and dual slack eigenvalue are respectively
\(3.4104\times10^{-13}\) and \(3.4102\times10^{-13}\). Dense SRM computations
using the full physical Gram matrix for the single-interval
model extend to \(n=50\), with candidate count \(M_n=1275\), over

\[
c\in\{0.5,0.8,0.9,0.95,0.99\}.
\]
For \(n>7\), only \(P_{\rm SRM}\) is computed. Calculations at larger \(n\)
therefore illustrate finite size SRM convergence; they neither estimate
\(P_{\rm opt}-P_{\rm SRM}\) nor independently establish SRM optimality.

\FloatBarrier

\section{Conclusion and discussion}
\label{sec:conclusion}

We have established exact asymptotic laws for minimum error localization of a
single returning quantum change interval under collective
measurements. For fixed \(c\) and a known length \(i\geq1\), the optimal
and SRM success probabilities under the uniform prior over
\(N=n-i+1\) translations converge to \(P_\infty(i,c^2)\) as
\(N\to\infty\). For fixed \(i\) and \(0<c<1\), their finite size gap
satisfies
\[
0\leq
P_{\rm opt}(G_{N,i})-P_{\rm SRM}(G_{N,i})
=
O_{i,c}\!\left(N^{-1/2}\right).
\]
This bound quantifies the convergence of the explicit SRM to the Bayes
optimum. At fixed \(c\), if \(i_n\to\infty\) and
\(N_n=n-i_n+1\to\infty\), both probabilities converge to \(p_1(c^2)\),
with no restriction on the relative growth of \(i_n\) and \(N_n\). 

When the interval length is unknown and the prior is uniform over all
\(M_n=n(n+1)/2\) nonempty intervals, both probabilities instead converge to
\(p_1(c)^2\) for fixed \(c\). Thus the SRM attains the collective Bayes
limit in all three localization regimes at fixed overlap. For an unknown interval length, the moving overlap analysis resolves the
nonuniform boundary at \(c=1\). The relevant scale is
\(\tau_n=n(1-c_n)(\log n)^2\). Theorem~\ref{thm:unknown-length-ultracritical} gives,
\[
M_nP_X(G_n(c_n))
=
F(\tau_n)
+
O_T\!\left(\frac{\log\log n}{\log n}\right).
\]
, \(X\in\{\mathrm{tr},\mathrm{SRM},\mathrm{opt}\}\). In a broader regime where the endpoint
hull dominates, Theorem~\ref{thm:unknown-length-hull-dominance} gives
\[
P_X(G_n(c_n))
\sim
p_1(c_n)^2
\]
whenever \(c_n\to1\), \(c_n<1\) eventually, and
\(n p_1(c_n)\to\infty\). This hull criterion applies across the inner,
continuum, and moving outer regimes.

The limits \(p_1(c^2)\) and \(p_1(c)^2\) reflect different geometries of the
hypothesis set. A growing interval of known length has a single translation
coordinate. Shifting the interval by one site changes one local state at each
boundary, producing the effective overlap \(c^2\). When the length is
unknown, the entry and return endpoints supply two bulk coordinates, and the
comparison kernel factorizes into one-dimensional kernels with parameter
\(c\), yielding \(p_1(c)^2\). The analytical strategy follows this geometry.
At fixed overlap, Toeplitz and circulant comparisons handle the known length
problem. For unknown length, local Toeplitz comparison, diagonal F{\o}lner
estimates, and perturbative Gram transfer control the physical Gram matrix on
the triangular endpoint domain despite its non-Toeplitz corrections. Near
\(c=1\), an exact decomposition by excitation number and a weighted hull
representation separate the vacuum and one excitation sectors from the
endpoint hull.

The no change extension preserves this asymptotic structure. In each model at
fixed overlap, adjoining a no change hypothesis with fixed prior
\(\pi_0\in(0,1)\) maps a conditional localization limit \(\Lambda\) to the
joint Bayes limit
\[
\pi_0+(1-\pi_0)\Lambda
\]
under the stated conditional prior on anomalous intervals. For every fixed
\(0\leq c<1\), the prior weighted SRM attains this limit both when the known
interval length grows and when the length is unknown. In these two settings,
maximum a posteriori postprocessing of the SRM outcomes is asymptotically
Bayes optimal for every fixed \(0\leq c\leq1\). In particular, at \(c=1\),
the unprocessed prior-weighted SRM tends to \(\pi_0^2\) as the number of
labels diverges, whereas MAP postprocessing recovers the Bayes limit
\(\pi_0\).

These results provide asymptotic benchmarks for unrestricted collective POVMs
on returning interval ensembles of known pure states. Local strategies have
been studied for permanent changes and isolated single anomalies
\cite{sentis2016,skotiniotis2024}. Whether local or sequential receivers
attain the collective limits derived here for a returning interval of fixed
length \(i\geq2\), a known length that grows with \(n\), or an unknown length
remains open. This structured localization problem differs from independently
drawn quantum sequence ensembles, for which optimal discrimination can factor
into fixed local measurements \cite{gupta2024sequence}. Natural extensions
include multiple returning intervals, mixed or correlated output states,
unknown or site-dependent anomaly states and overlaps, imperfect returns, and
restricted measurement architectures. The central question is whether the
geometry of the two endpoints continues to control the asymptotics beyond the
calibrated pure state model.

\section*{Author contributions}

Xu Chen conceived the initial idea, developed the physical model, and carried out the theoretical analysis. Xue Ma performed the numerical calculations and prepared the figures. Both authors contributed writing and revision of the manuscript. LLM Models were used for language polishing and cross checking numerical results. Both authors reviewed the  outputs and take full responsibility for the scientific content and the final manuscript.

\section*{Acknowledgements}

This work was supported by Hebei University of Science and Technology under
Starting Grant No.~81/1181298.

\bibliographystyle{quantum}
\bibliography{references}

\end{document}